\documentclass[thesis,11pt]{amsart}%
\usepackage{amssymb}
\usepackage{amsmath}
\usepackage{mathrsfs}
\usepackage{amsfonts}
\usepackage{graphicx}
\usepackage[left=1.5in, right=1.5in, top=1in, bottom=1in, includefoot, headheight=13.6pt]{geometry}
\usepackage[all]{xy}
\newtheorem*{theorem}{Theorem}

\theoremstyle{plain}

\newtheorem{proposition}{Proposition}
\newtheorem*{remark}{Remark}

\numberwithin{equation}{section}
\begin{document}
\title[\emph{Quantization of Time}]{The Stochastic Representation of Hamiltonian Dynamics and The Quantization of Time}
\author{M.~F.~Brown}
\address{Mathematics Department, University of Nottingham\\
NG7 2RD, UK.}
\email{pmxmb1@nottingham.ac.uk}
\date{June 2012}
%\subjclass[2000]{Primary 60H99; Secondary 60G99.}
%\keywords{Quantum Dynamics, Time Quantization, Foundations of Physics.}
\maketitle
\begin{center}
\emph{Thankyou Slava for your teachings}
\end{center}
\begin{abstract}
Here it is shown that the unitary dynamics of a quantum object may be obtained as the expectation of a counting process of object-clock interactions. Such a stochastic process arises from the quantization of the clock, and this is derived naturally from the matrix-algebra representation \cite{Be92b} of the nilpotent Newton-Leibniz differential time increment, $\mathrm{d}t$. Following \cite{Be00a} it is also shown that the object-clock interaction dynamics is unitarily equivalent to a pseudo-selfadjoint Schr\"odinger past-future boundary value problem.
%\PACS{03.65.Ta} % \and PACS code2 \and more}
%\subclass{81S25 \and 81S20 \and 81P05 \and 81P15 \and 83C99}
\end{abstract}
\tableofcontents
\section{Introduction}
Quantum stochastic evolution is still an unfamiliar territory to many physicists but provides a powerful analytic tool for the study of open quantum systems. However, even among those familiar with the quantum stochastic calculus of Hudson and Parthasarathy \cite{HudP84} the more general, and more rigorous, Belavkin formalism \cite{Be91,Be92,Be92b} still remains unknown. Here it is intended to introduce the Belavkin formalism of quantum stochastic calculus as a fundamental concept of physics by focusing on the Schr\"odinger equation of a deterministic Hamiltonian dynamics. This is achieved by virtue of the quantization of time, which is an idea that arises from the Belavkin formalism alone and is lost in the Husdon-Parthasarathy formalism.

Quantum stochastic calculus is an algebraic approach to stochastic analysis that is applicable to all stochastic processes whether quantum or not. It was constructed using Hudson's and Parthsarathy's quantum It\^o algebra $\mathfrak{A}$ of stochastic increments \cite{HudP84} having a basis of four fundamental increments; these are $\mathrm{d}t,\mathrm{d}a,\mathrm{d}a^\ast,\mathrm{d}n$ that are respectively called the Newton-Leibniz, annihilation, creation, and counting differential increments. One may bare in mind that the Belavkin formalism arose from Belavkin's construction of the representing vector space of a quantum It\^o algebra. This turned out not to be a Hilbert space but remarkably a Minkowski-Hilbert space as we shall see more explicitly in the next chapter. The general theory of quantum stochastic calculus is beyond the scope of this paper, but one should note that the mathematical tool for the study of quantum stochastic evolution is Fock space.

Indeed it seems strange that one should be able to come to an understanding of the Belavkin formalism of quantum stochastic calculus through the Schr\"odinger equation which describes a deterministic evolution, not a stochastic one. However, as it shall be shown, such Hamiltonian dynamics of a closed quantum system represented in a Hilbert space $\mathfrak{h}$ may be dilated to a discrete object-clock stochastic interaction dynamics, a counting process, in the Hilbert product $\mathbb{H}:=\mathfrak{h}\otimes\mathbb{F}$ where $\mathbb{F}$ is the Minkowski-Fock space of the clock corresponding to the quantization of time. As one may now expect, the Hamiltonian dynamics described by the the Schr\"odinger equation in $\mathfrak{h}$ is obtained by tracing out over the Minkowski-Fock space of the `quantum clock'. It will become apparent that such `tracing out' is a rigorous approach to the Feynman path integral.

The purpose of this work is thus twofold. An understanding of this stochastic dilation of deterministic evolution provides a solid foundation for understanding the Belavkin formalism of the general quantum stochastic calculus since the general theory corresponds to the introduction of an additional `noise' degree of freedom to the quantum clock. As indeed it is true that, for example, such evolution as quantum Brownian motion may be given in the Belavkin formalism as the conditional expectation of a counting process. This runs deeply into the analysis of operators in Fock space, for it means that such operators may be obtained as the conditional expectation of Block-diagonal operators in a bigger Minkowski-Fock space of a stochastic system plus clock. And the purpose of this work is also to give rise to a new interpretation of the deterministic dynamics of a closed quantum system by performing a `microscopic dilation' of the system that opens up the evolution by separating the system and the time.

\section{The Algebraic Realization of Differential Calculus}
After the initial construction of the quantum  It\^o algebra in 1984 by Hudson and Parthasarathy \cite{HudP84}, Belavkin developed a more rigorous $\star$-algebraic formulation of this quantum It\^o calculus \cite{Be91,Be92,Be92b} that begins by defining an involution on the quantum It\^o algebra that is denoted by $\star$. The involution leaves the basis elements $\mathrm{d}t$ and $\mathrm{d}n$ invariant but swaps the increments of creation and annihilation, and the introduction of this involution \cite{Be92b} was the first step that had to be made in order to obtain a Minkowski space representation of the clock.

With attention now restricted to deterministic evolution the Newton-Leibniz time differential $\mathrm{d}t$ may be understood as the basis of a \emph{nilpotent $\star$-algebra}, that is  a one-dimensional sub-algebra $\mathfrak{a}\subset\mathfrak{A}$ of a general  $\star$-algebra $\mathfrak{A}$. The elements of $\mathfrak{a}$ have the form
\[a=\alpha\mathrm{d}t\]
 where $\mathrm{d}t$ is the single basis element of the algebra, and $\alpha$ is called the coefficient of $a$ and defines a deterministic derivative.
 The involution in $\mathfrak{a}$ is given by $\star:a\mapsto a^\star$ and leaves the base $\mathrm{d}t$ invariant, and may therefore be given by the involution of the coefficients  such that $a^\star=\overline{\alpha}\mathrm{d}t$ when $\alpha\in\mathbb{C}$, where $\overline{\alpha}$ is complex conjugation. Moreover,  $\mathfrak{a}$ may be  given as a linear isomorphism of the space of coefficients.

 The nilpotent $\star$-algebra is generated by the usual operation of addition but also the nilpotent Newton-Leibniz product $\bullet$ corresponding to the nilpotent property $\mathrm{d}t\mathrm{d}t=0$, such that
 \[
 a+a'=\big(\alpha+\alpha'\big)\mathrm{d}t,\quad aa'=0=\alpha\bullet \alpha'
 \]
 in $\mathfrak{a}$,  for all $a,a'\in\mathfrak{a}$.
 The $\star$-algebras are non-unital, but one may unitalize $\mathfrak{a}$ by adding the unit $\boldsymbol{1}\notin\mathfrak{a}$,  such that \[\boldsymbol{1} a=a=a\boldsymbol{1}\;\;\;\forall\;\;a\in\mathfrak{a}\]
 forming a $\star$-group $\mathfrak{g}=\boldsymbol{1}+\mathfrak{a}\equiv\mathfrak{a}_{\boldsymbol{1}}$ in one extra dimension (since $\boldsymbol{1}$ is not of the form $\alpha\mathrm{d}t$) and the involution in this group is defined on the elements $\boldsymbol{1}+a$ by $(\boldsymbol{1}+a)^\star=\boldsymbol{1}+a^\star$.
The construction of $\mathfrak{g}$ realizes the associative, but not distributive, \emph{monoidal} product $\cdot$  defined in the algebra $\mathfrak{a}$, by virtue of the group $\mathfrak{g}$, as \[\boldsymbol{1}+a\cdot a':=(\boldsymbol{1}+a)(\boldsymbol{1}+a')\]
  \cite{Be92b}. Since the $\star$-algebra product is  distributive we find that \[
  a\cdot a'\equiv a+a a'+a'=a+a'\]
   for all $a,a'\in\mathfrak{a}$, having the  unit $0\in a$ such that $0\cdot a=a=a\cdot 0$. One may wish to note that when working with the {general} $\star$-algebra $\mathfrak{A}$, containing also the algebraic realization of Wiener and Poisson increments,  the associative $\star$-algebra product  does not have the nilpotent Newton-Leibniz form but is the more general {It\^o product}, and this describes the differential It\^o correction of stochastic calculus.  Further, the unitalization $\mathfrak{A}_{\boldsymbol{1}}$ does not form a $\star$-group but  only $\star$-semigroup.

It is common knowledge in functional analysis that Gelfand, Naimark, and Segal, proved that every $C^\ast$-algebra can be represented as an operator algebra in a Hilbert space; that is  the {GNS construction}. However, proceeding in this manner  Belavkin discovered \cite{Be92,Be92b} that in the more general setting of It\^o $\star$-algebras, every such $\star$-algebra $\mathfrak{A}$ could be represented in a \emph{pseudo}-Hilbert space. Further, such pseudo-Hilbert space did not have arbitrary pseudo-metric, but a Minkowski metric.

To this end we shall investigate this \emph{Belavkin representation} of the nilpotent $\star$-algebra $\mathfrak{a}\subset\mathfrak{A}$ of the Newton-Leibniz calculus, and we shall come to a new understanding of the quantum dynamics generated by Hamiltonian operators; such is called  \emph{the stochastic representation of Hamiltonian dynamics}. It may even be argued that this stochastic representation of Newton-Leibniz calculus provides  microscopic interpretation of deterministic dynamics.

\section{Quantization of The Clock}
The nilpotent $\star$-algebra $\mathfrak{a}$ has a matrix representation $\mathfrak{d}$ that is a one-dimensional sub-algebra of the $2\times2$ matrix algebra $\mathcal{M}_2$. This is given by the $\star$-homomorphism $\pi:\mathfrak{a}\mapsto\mathfrak{d}$, that is
\begin{equation}
\pi(a'a^\star)=\pi(a')\pi(a)^\ddag
 \end{equation}
for all $a',a\in\mathfrak{a}$, and $\pi$ is extended onto the group $\mathfrak{g}$ as  $\pi(\boldsymbol{1})=\mathbf{I}$ where $\mathbf{I}$ is matrix identity in $\mathcal{M}_2$.
 Notice that we have a $\ddag$-involution in $\mathfrak{d}$ representing the $\star$-involution in $\mathfrak{a}$, and this is the pseudo-involution  given by a Minkowski metric $\boldsymbol{q}$ so that
\[\pi(a)^\ddag:=\boldsymbol{q}^{-1}\pi(a)^\ast \boldsymbol{q}^\ast=\pi(a^\star)
\]
with $\pi(\mathrm{d}t)^\ddag=\pi(\mathrm{d}t)$, where the $\ast$-involution is usual conjugation in $\mathcal{M}_2$ of matrix transposition and  generic involution in the space of coefficients (that is complex conjugation if the coefficients are complex).

In the canonical upper-triangular form \cite{Be92b} the elements of the $\ddag$-algebra $\mathfrak{d}$ (that is algebra with the Minkowski involution denoted by $\ddag$) are given explicitly, with their $\ddag$-adjoint, as
\[
 \pi(a)=\left[
          \begin{array}{cc}
            0 & \alpha \\
            0 & 0 \\
          \end{array}
        \right]\equiv\alpha\otimes\pi(\mathrm{d}t),\quad
        \pi(a)^\ddag=\left[
          \begin{array}{cc}
            0 & \alpha^\ast \\
            0 & 0 \\
          \end{array}
        \right]\equiv\alpha^\ast\otimes\pi(\mathrm{d}t)
 \]
 with respect to the elements $a=\alpha\mathrm{d}t$ in $\mathfrak{a}$, such that the Minkowski metric has the non-diagonal form
 \[
 \boldsymbol{q}=\left[
          \begin{array}{cc}
            0 & 1 \\
            1 & 0 \\
          \end{array}
        \right]\equiv\sigma_1.
 \]
Notice that the representing pseudo-Hilbert space of $\mathfrak{a}$ (called Minkowski-Hilbert space) always has these two degrees of freedom independent of the space of coefficients $\alpha$,  and it is important to understand that this is a consequence of the $2\times2$ matrix representation \begin{equation}\pi(\mathrm{d}t)=\left[
                     \begin{array}{cc}
                       0 & 1 \\
                       0 & 0 \\
                     \end{array}
                   \right]:=\mathbf{d}
\end{equation} of the differential time increment $\mathrm{d}t$ that forms the single basis element of the Newton-Leibniz algebra. These degrees of freedom shall be referred to as the
\emph{temporal-spin} of the clock having the  \emph{past} and \emph{future} null-states given respectively as
\[
\left[
    \begin{array}{c}
      1 \\
      0 \\
    \end{array}
  \right]\;\;\textup{and}\;\; \left[
    \begin{array}{c}
      0 \\
      1 \\
    \end{array}
  \right].
\]
Note in particular that this algebraic realization of Newton-Leibniz calculus allows us to represent a basic differential increment of time as a matrix operation transforming future into past
\[
\left[
           \begin{array}{cc}
             0 & 1 \\
             0 & 0 \\
           \end{array}
         \right]:\left[
    \begin{array}{c}
      0 \\
      1 \\
    \end{array}
  \right]\mapsto\left[
    \begin{array}{c}
      1 \\
      0 \\
    \end{array}
  \right],
\]
and so begins the algebraic realization of calculus.
\\
\linebreak
In the simplest case of scalar coefficients we have $\mathfrak{a}=\mathbb{C}\mathrm{d}t$, and this scalar $\star$-algebra is represented in the Minkowski-Hilbert space $\Bbbk=(\mathbb{C}^2,\sigma_1)$ \cite{thesis}
of column vectors
\[\xi=\left[
            \begin{array}{c}
              \xi^- \\
              \xi^+ \\
            \end{array}
          \right],
\] with the $\ddag$-adjoint rows
\[\xi^\ddag=\big(\xi_-,\xi_+\big)=\xi^\ast\sigma_1 \]
in the dual space $\Bbbk^\ddag\cong\Bbbk$ of linear functionals $\xi^\ddag:\Bbbk\rightarrow\mathbb{C}$,
where $\xi_\pm=\overline{\xi^\mp}$ and $\overline{\xi}$ is complex conjugation of $\xi$, and we define the $\ddag$-norm on $\Bbbk$  as
\begin{equation}
\|\xi\|=\Big(\xi_-\xi^-+\xi_+\xi^+\Big)^{\frac{1}{2}}\equiv\Big( \xi^\ddag\xi\Big)^{\frac{1}{2}}.\label{pip}
\end{equation}
Thus we may consider states $\rho_\xi$ on the $\star$-algebra $\mathbb{C}\mathrm{d}t$ given by a $\xi\in\Bbbk$ with $\xi_-\xi^+=|\xi^+|^2_{\mathbb{C}}=1$, where $|\cdot|_{\mathbb{C}}$ is norm in $\mathbb{C}$, such that
\[
\rho_\xi(a)=\xi^\ddag\alpha\pi(\mathrm{d}t)\xi=\alpha
\]
where $a=\alpha\mathrm{d}t$. These states may be linearly extended onto the $\star$-group $\mathfrak{g}$ by defining their action on the unit $\boldsymbol{1}$. We shall see that in order to have a pure state in the second quantization of the clock (that is an exponentiation) the states $\rho_\xi$ must be null such that $\rho_\xi(\boldsymbol{1})=0$. Thus purity is achieved if $\xi$ is a future state.

  So we have just considered the derivative $\alpha$ at some arbitrary fixed time $x=t$, but we would now like to consider a variable coordinate of the clock, $x\in\mathbb{R}_+$.  Thus consider a family $\{\mathfrak{a}_x\}$ of scalar $\star$-algebras  $\mathfrak{a}_x\equiv\mathbb{C}\mathrm{d}x$ over $\mathbb{R}_+$,
and a family of representing  Minkowski-Hilbert spaces $\{\Bbbk_x\}$ for the differentials $a(x)=\alpha(x)\mathrm{d}x\in\mathfrak{a}_x$, with $\Bbbk_x=(\mathbb{C}^2,\sigma_1)$ at each $x\in\mathbb{R}_+$. Now we define the nilpotent $\star$-algebra $\mathfrak{a}=L^1(\mathbb{R}_+)\mathrm{d}{t}$ with elements of the form
\[
\alpha\mathrm{d}{t}:x\mapsto \alpha(x)\mathrm{d}x\in\mathfrak{a}_x
\]
where $\alpha$ is in the space $L^1(\mathbb{R}_+)$ of Lebesgue integrable functions over $\mathbb{R}_+$.

Belavkin's representing Minkowski-Hilbert space for this Newton-Leibniz algebra has the general form \cite{Be92,Be92b,thesis}
\begin{equation}
\Bbbk=\big(L^1(\mathbb{R}_+)\oplus L^\infty(\mathbb{R}_+),\sigma_1\big)\label{phs}
\end{equation}
where $L^\infty(\mathbb{R}_+)$ is the space of essentially bounded functions over $\mathbb{R}_+$,
and the $\mathbb{R}$-valued pseudo-norm is given on the vector functions $\xi\in\Bbbk$ as
\begin{equation}
\|\xi\|=\bigg(\int_0^\infty\|\xi(x)\|^2\mathrm{d}x\bigg)^{\frac{1}{2}}
\end{equation}
where $\|\xi(x)\|$ is the norm of $\xi(x)$ in $\Bbbk_x\cong\mathbb{C}^2$ given by (\ref{pip}).
Notice that  the nilpotent  product $\bullet$ of two integrable functions is  zero  such that $L^1(\mathbb{R}_+)$   forms a  nilpotent $\star$-algebra
although it is  not a $\ast$-algebra with respect to the usual point-wise product of functions over $\mathbb{R}_+$.
\\
\linebreak
The pseudo-Hilbert space (\ref{phs}) may be realized from the requirement that the integrability of $\alpha$ must be preserved in the compression of $\boldsymbol{\alpha}=\alpha\otimes\pi(\mathrm{d}t)$ such that
\[
\big|\xi^\ddag \boldsymbol{\alpha}\xi\big|_{\mathbb{C}}<\infty\;\;\forall\;\;\xi\in\Bbbk\]
where $|\cdot|_{\mathbb{C}}$ is the norm in $\mathbb{C}$,
   and this means that the components $\xi^-,\xi^+$ of $\xi$ cannot be in $L^2(\mathbb{R}_+)$. However, we can realize  $\xi^-\in L^1(\mathbb{R}_+)$ and $\xi^+\in L^\infty(\mathbb{R}_+)$, such that
    \[
    \|\xi\|^2=\sum_{\kappa=-,+}\langle\xi_\kappa,\xi^\kappa\rangle\equiv \sum_{\kappa=-,+}\int_{\mathbb{R}_+}\xi_\kappa(x)\xi^\kappa(x)\mathrm{d}x
    \]
    with $\xi_-\in L^\infty(\mathbb{R}_+)$ and $\xi_+\in L^1(\mathbb{R}_+)$,
     such that $\langle\xi_\kappa,\xi^\kappa\rangle$, $\kappa=-,+$, is the dual pairing of $L^1(\mathbb{R}_+)$ and $L^\infty(\mathbb{R}_+)$. Note that in the context of integration $\mathrm{d}x$ is the Lebesgue measure on $\mathbb{R}_+$. This
  preserves the integrability of $\alpha$ in the pseudo-inner product as
\begin{equation}
\xi^\ddag \boldsymbol{\alpha}\xi=(\xi_-,\xi_+)\left[
                                       \begin{array}{cc}
                                         0 & \alpha \\
                                         0 & 0 \\
                                       \end{array}
                                     \right]\left[
                                              \begin{array}{c}
                                                \xi^- \\
                                                \xi^+ \\
                                              \end{array}
                                            \right]
                                           \equiv \langle \xi_-,\alpha\xi^+\rangle,
\end{equation}
where $\langle \xi_-,\alpha\xi^+\rangle$ is the dual pairing of $\xi_-\in L^\infty(\mathbb{R}_+)$ with $\alpha\xi^+\in L^1(\mathbb{R}_+)$, noting that the point-wise product $[\alpha\xi^+](x)=\alpha(x)\xi^+(x)$, of an essentially bounded function with an integrable function, is integrable. Also notice that if $\alpha>0$  then  $\xi^\ddag \boldsymbol{\alpha}\xi>0$ for all $\xi\in\Bbbk$ since $\xi^+=\xi_-^\ast$ such that $\langle \xi_-,\alpha\xi^+\rangle=\langle{\xi^+}^\ast\xi^+,\alpha\rangle$, where $[{\xi^+}^\ast\xi^+](x)=\overline{\xi^+(x)}\xi^+(x)$.
Thus we shall conclude this section with a consideration of the $\star$-algebra norm.

Since $\mathfrak{a}\cong L^1(\mathbb{R}_+)$ (linearly isomorphic) we simply define the norm on $\mathfrak{a}$ as \[\|\alpha \mathrm{d}\boldsymbol{t}\|=\|\alpha\|_1\equiv\big\langle1,|\alpha|\big\rangle,\quad1\in L^\infty(\mathbb{R}_+) \] as it was done in \cite{Be92,Be92b}, where $|\alpha|(x):=|\alpha(x)|_{\mathbb{C}}$,
and this norm may be obtained from basic vectors of the form ${\xi_\nu^\mu}=e^{-\mathrm{i}\mu \boldsymbol{t}}\xi_\nu$,
where
\begin{equation}
\xi_\nu=\left[
          \begin{array}{c}
            \nu \\
            1 \\
          \end{array}
        \right]\label{basic}
\end{equation}
 with $\nu\in L^1(\mathbb{R}_+)$ and $1\in L^\infty(\mathbb{R}_+)$, and $\boldsymbol{t}$ is the clock-coordinate operator $[\boldsymbol{t}\xi](x)=t(x)\xi(x)$, such that
\begin{equation}
{\xi^\mu_\nu}^\ddag\big[|\alpha|\otimes\pi(\mathrm{d}t)\big]\xi^\mu_\nu= \big\langle 1,|\alpha|\big\rangle\equiv \int_0^\infty |\alpha(x)|_{\mathbb{C}}\mathrm{d}x
\end{equation}
  for all $\nu\in L^1(\mathbb{R}_+)$ and $\mu\in\mathbb{R}$. The parameters $\mu$ and $\nu$ of the vector functions ${\xi_\nu^\mu}=e^{-\mathrm{i}\mu \boldsymbol{t}}\xi_\nu$ are non-trivial. We shall see in some detail below that $\nu$ describes a the interaction frequency of the dilated object-clock dynamics. We shall also see below that the parameter $\mu$ is a momentum that generates a free dynamics of the clock.

\section{The Stochastic Representation of Hamiltonian Dynamics}
So we have established that the nilpotent $\star$-algebra $\mathfrak{a}=L^1(\mathbb{R}_+) \mathrm{d}\boldsymbol{t}$ has matrix representation $\mathfrak{d}=L^1(\mathbb{R}_+)\otimes\mathbf{d}$, where $\mathbf{d}=\pi(\mathrm{d}t)$, in the Minkowski-Hilbert space $\Bbbk=\big(L^1(\mathbb{R}_+)\oplus L^\infty(\mathbb{R}_+),\sigma_1\big)$; that is the quantization of the clock obtained from the matrix representation of the Newton-Leibniz time increment. Now we would like to consider an arbitrary quantum system composed with this quantum clock; we shall refer to such quantum system as the object.
Thus we shall consider the algebraic tensor product $\mathcal{B}(\mathfrak{h})\bar{\otimes}\mathfrak{a}$ of $\mathfrak{a}$ with a  $C^\ast$-algebra $\mathcal{B}(\mathfrak{h})$ of bounded linear operators in a Hilbert space $\mathfrak{h}$. This forms a composite $\star$-algebra with pseudo-involution   given on the separable elements by the map  $\star:B\otimes a\mapsto B^\ast\otimes a^\star$.
We shall denote by $\boldsymbol{\gamma}=G\otimes\mathbf{d}$ the elements of the  $\ddag$-matrix algebra $\mathcal{B}(\mathfrak{h})\bar{\otimes}\mathfrak{d}$, considered as the maps
\[\boldsymbol{\gamma}:x\mapsto G(x)\otimes\mathbf{d},\]
 with $G(x)$ in $\mathcal{B}(\mathfrak{h})$ for almost all  $x\in\mathbb{R}_+$ as $G\in \mathcal{B}(\mathfrak{h})\bar{\otimes} L^1(\mathbb{R}_+)$.

The Hilbert space  $\mathfrak{h}$ is introduced here as a space of square-integrable functions $\eta$ defining the probability amplitudes of the internal degrees of freedom of a quantum system. Generally the Hilbert space $\mathfrak{h}$ is infinite dimensional, but need not be; here the choice is arbitrary. The task now before us is to present the underlying mechanism, as discovered in the more general setting by Belavkin, from which this object's deterministic unitary evolution may be derived.

 We begin with the compound Minkowski-Hilbert space $\mathfrak{h}\otimes\Bbbk$, that is the Hilbert product of the quantum object's Hilbert space $\mathfrak{h}$  with the Minkowski-Hilbert space $\Bbbk$ of the clock. The object is considered as an open sub-system in this compound space, and we shall see that a unitary dynamics of the object in $\mathfrak{h}$ may be obtained  as an expectation of a sequential interaction dynamics in Minkowski-Hilbert space. This sequential interaction dynamics is the entanglement of the object with the clock at arbitrary interaction times $\tau=\{x_1<\ldots<x_n\}$, $x_i\in\mathbb{R}_+$, and this time-ordered sequence of interactions is described by the equations of the  quantum measurement dynamics introduced in \cite{Be93}; but in contrast, here we work with Minkowski-Hilbert space.
 As a consequence of such quantum measurement interpretation, the action of the  Hamiltonian in $\mathfrak{h}$ may be understood as the indirect observation (measurement) of the passing of the object's time. Further, it may be inferred  that an increment of time  is the consequence of the observation.
  \\
\linebreak
Interaction operators in a closed system must be unitary, or in this case pseudo-unitary. So first we shall consider the general  pseudo-unitary operators $\boldsymbol{g}^\ddag=\boldsymbol{g}^{-1}$ in $\mathfrak{h}\otimes\Bbbk$, having the form
\[
\boldsymbol{g}=\left[
  \begin{array}{cc}
    Z^{-1} & G \\
    0 & Z^{\ast} \\
  \end{array}
\right]
\]
where $Z$ is invertible in $\mathcal{B}(\mathfrak{h})\bar{\otimes} L^\infty(\mathbb{R}_+)$ and $(ZG)^\ast=-ZG$, such that the action of the operators $\boldsymbol{g}$ on vectors functions $\psi\in\mathfrak{h}\otimes\Bbbk$ has the diagonal form \[[\boldsymbol{g}\psi](x)=\boldsymbol{g}(x)\psi(x).\] With calculus in mind one
  must insist that $Z= I$ such that $\boldsymbol{g}$ is an element of the $\ddag$-group $\mathcal{B}(\mathfrak{h})\bar{\otimes}\mathfrak{m}$ where $\mathfrak{m}:=\pi(\boldsymbol{1})+\mathfrak{d}=\pi(\mathfrak{g})$, where $\pi(\boldsymbol{1})$ is the identity $\mathbf{I}$ in $\mathcal{M}_2$. To see this, consider  an arbitrary operator $V\in\mathcal{B}(\mathfrak{h})$ whose incremental evolution may be given by the  $\star$-algebraic differential equation
\[
\boldsymbol{1} V+\mathrm{d}V=\boldsymbol{G}V
\]
 corresponding to the homogeneous (logarithmic) derivative $\mathrm{d}V=\boldsymbol{L}V$ in $\mathcal{B}(\mathfrak{h})\bar{\otimes}\mathfrak{a}$ where $\boldsymbol{L}=\boldsymbol{G}-\boldsymbol{1}$, with $\pi(\boldsymbol{G})=\boldsymbol{g}$.
Thus we shall denote by
\begin{equation}
\boldsymbol{\sigma}=\left[
             \begin{array}{cc}
               I & G \\
               0 & I \\
             \end{array}
           \right]\equiv \mathbf{I}+\boldsymbol{\gamma},\quad \boldsymbol{\gamma}=G\otimes\mathbf{d},\label{gamma}
\end{equation}
the canonical pseudo-unitary operators of the object-clock interactions. Notice that this $\ddag$-unitarity of $\boldsymbol{\sigma}$ imposes the constraint $G^\ast=-G$ in $\mathcal{B}(\mathfrak{h})\bar{\otimes} L^1(\mathbb{R}_+)$, which is the property that we require for a generator of unitary evolution in $\mathfrak{h}$ defining the object Hamiltonian $H(x)=\mathrm{i}G(x)$ at almost all $x\in\mathbb{R}_+$.

Notice that if we impose the constraint that the coefficients of $\boldsymbol{g}$ be $\mathbb{R}$-valued then the $\ddag$-unitary operators $\boldsymbol{g}^\ddag=\boldsymbol{g}^{-1}$ are none other than the real Lorentz transformations in their diagonal representation.
\\
\linebreak
%Before we consider the initial configuration of the object-clock system note that the action of the operators $\boldsymbol{\sigma}$ is given on the vectors $\xi\in\Bbbk$ as the point-wise product $[\boldsymbol{\sigma}\xi](x)=\boldsymbol{\sigma}(x)\xi(x)$.
Now, we consider the object to  be prepared in an arbitrary, but normalized, initial state $\eta\in\mathfrak{h}$, but we shall suppose that the clock is in an initial null state $\xi_0\equiv\xi_\emptyset$ (\ref{basic}), such that the temporal spin of the clock is reckoned to be in the future-state at all $x$ in $\mathbb{R}_+$. This means that the \emph{input state} is future.
Then a single  object-clock interaction may be given at an arbitrary point $x$ as the map
\[
\boldsymbol{\sigma}(x)=\left[
             \begin{array}{cc}
               I & G(x) \\
               0 & I \\
             \end{array}
           \right]:\left[
                     \begin{array}{c}
                       0 \\
                       \eta \\
                     \end{array}
                   \right]\mapsto\left[
                     \begin{array}{c}
                       G(x)\eta \\
                       \eta \\
                     \end{array}
                   \right]\label{pui}
\]
entangling the object and clock since the $\ddag$-unitary operator $\boldsymbol{\sigma}(x)$ is not separable. Further, the differential increment of the system, given by the action of the operator $\boldsymbol{\gamma}(x)$ (\ref{gamma}), is the transformation
\begin{equation}
\eta\otimes\left[
                     \begin{array}{c}
                      0 \\
                       1 \\
                     \end{array}
                   \right]\mapsto[G(x)\eta]\otimes\left[
                     \begin{array}{c}
                       1 \\
                       0 \\
                     \end{array}
                   \right],\label{act}
\end{equation}
such that the action of $G(x)$ on $\eta$ is an indirect measurement of the transition of future into past. Notice that if the input state is past then the interaction leaves the object state is unchanged.

The single jump equation for this pseudo-unitary interaction  is given by an It\^o-Schr\"odinger equation \cite{Be93} of the form
\begin{equation}
\mathrm{d}\psi_{t}(x)=\boldsymbol{\gamma}(x)\psi_{t}(x)\mathrm{d}1_{t}(x), \quad \psi_0(x)=\eta\otimes\xi_\emptyset(x),\label{1jump}
\end{equation}
where $x\in\mathbb{R}_+$ is the variable clock coordinate, $\boldsymbol{\gamma}(x)=\boldsymbol{\sigma}(x)-\mathbf{I}$, and $1_t(x)=1$ if $x<t$ and is otherwise zero, with $\mathrm{d}1_t(x)=1_{t+\mathrm{d}t}(x)-1_t(x)$. The parameter $t$ is the evolution parameter of this  interaction dynamics, and (\ref{1jump}) has the solution
\begin{equation}
\psi_t(x)=\boldsymbol{\sigma}_t(x)\psi_0(x),\quad \boldsymbol{\sigma}_t(x)=\boldsymbol{\gamma}(x)1_t(x)+ \mathbf{I},
\end{equation}
where $\psi_t(x)\in\mathfrak{h}\otimes\mathbb{C}^2$ with $\psi_t\in\mathfrak{h}\otimes\Bbbk$.
\begin{remark}
Notice that we may consider basic separable elements in $\mathcal{B}(\mathfrak{h})\bar{\otimes} L^1(\mathbb{R}_+)$ of the form $G=\breve{G}\otimes\nu$, where $\breve{G}\in\mathcal{B}(\mathfrak{h})$ and $\nu\in L^1(\mathbb{R}_+)$, such that $G(x)=\breve{G}\nu(x)$. In particular, this allows us to write the action of $\boldsymbol{\sigma}$ as
$\boldsymbol{\sigma}[\eta\otimes\xi_\emptyset]=[\breve{G}\eta]\otimes\xi_\nu$ $\big($where $\xi_\nu$ is given by (\ref{basic})$\big)$,  with $\nu(x)\in\mathbb{R}$ preserving the anti self-adjointness $\breve{G}^\ast=-\breve{G}$. Consider for example $\nu=1_t$.
\end{remark}
Now we shall consider a spontaneous process of these pseudo-scattering interactions between the object and the clock. To do this we must consider the second quantization $\mathbb{F}=\Gamma(\Bbbk)$ of the clock, that is given with respect to the pseudo-Fock norm below as the closed $\mathbb{C}$-linear span of the Hilbert product functions
\begin{equation}
\xi^\otimes:\tau\mapsto\underset{{x\in\tau}}{\otimes}\xi(x)\label{tens}
\end{equation}
given on  the ordered countable finite sets  $\tau=\{x_1<\cdots<x_{|\tau|}\}$, where $|\tau|=n(\tau)\in\{0\}\cup\mathbb{N}$ is the cardinality of this chain $\tau$ defining the number of interaction coordinates, called \emph{particles}, in the object's future. The space of all such finite chains is the disjoint union $\mathcal{X}=\sqcup_{n=0}^\infty\mathcal{X}_n$, where the space $\mathcal{X}_n\subset\mathbb{R}_+^n$ is the $n$-simplex over $\mathbb{R}_+$ of $n$-particle chains $\tau_n=\{x_1<\ldots<x_n\}$. The pseudo-norm on $\mathbb{F}$ is given on the product functions as
\begin{equation}
\|\xi^\otimes\|=\bigg(\int_{\mathcal{X}}\|\xi^\otimes(\tau)\|^2\mathrm{d}\tau \bigg)^{\frac{1}{2}} \equiv\bigg(\sum_{n=0}^\infty\int_{\mathcal{X}_n} \prod_{x\in\tau_n}\|\xi(x)\|^2\mathrm{d}\tau_n\bigg)^{\frac{1}{2}}= \exp\big\{\tfrac{\xi^\ddag\xi}{2}\big\},
\end{equation}
with respect to the $\ddag$-norm $\|\xi(x)\|$ on $\Bbbk_x$, where $\prod_{x\in\tau}\|\xi(x)\|=\|\xi^\otimes(\tau)\|$, and
\[
\int_{\mathcal{X}_n}f(\tau_n)\mathrm{d}\tau_n=\underset{0<x_1<\ldots<x_n<\infty}{\int\ldots\int} f(x_1,\ldots,x_n)\mathrm{d}x_1\ldots\mathrm{d}x_n,\quad f\in L^1(\mathcal{X}_n).
\]
This norm is given with respect to the Lebesgue measure $\mathrm{d}x$ on $\mathbb{R}_+$, and the atomic measure $\mathrm{d}\emptyset=1$ on the only atomic point $\tau_0=\emptyset$. This construction is called the Guichardet-Fock quantization of the clock \cite{Be92b,BelB10,thesis,Gui72} and $\mathbb{F}$ ia called Minkowski-Fock space,
and with respect to the pseudo-Fock norm we may define the \emph{coherent states} as the normalized product vectors $\xi^\otimes\exp\big\{-\frac{\xi^\ddag\xi}{2}\big\}$.

It should also be brought to one's attention  that the second quantization functor $\Gamma$ admits the identification
\begin{equation}
\Gamma\big(L^1(\mathbb{R}_+)\oplus L^\infty(\mathbb{R}_+)\big) =\Gamma\big(L^1(\mathbb{R}_+)\big)\otimes\Gamma\big(L^\infty(\mathbb{R}_+)\big),\label{fun}
\end{equation}
and the pseudo-Fock norm may be given with respect to this identification, for $\xi^\ddag=(\xi_-,\xi_+)$, as
\[
\|\xi_-^\otimes\otimes \xi_+^\otimes\|^2=\int_{\mathcal{X}}\int_{\mathcal{X}}\|\xi_-^\otimes(\upsilon)\xi_+^\otimes(\omega)\|^2 \mathrm{d}\upsilon\mathrm{d}\omega\equiv|\langle \xi_-^\otimes,\overline{\xi}_+^\otimes\rangle |_{\mathbb{C}}^2,
\]
with respect to the norms $\|\chi(\upsilon,\omega)\|^2=\chi^\ast(\omega,\upsilon)\chi(\upsilon,\omega)$, where $\chi^\ast(\omega,\upsilon)=\chi^\ddag(\upsilon,\omega)$, and the bilinear paring $\langle\cdot,\cdot\rangle$ of $L^\infty(\mathcal{X})=\Gamma\big(L^\infty(\mathbb{R}_+)\big)$ and $L^1(\mathcal{X})=\Gamma\big(L^1(\mathbb{R}_+)\big)$, where $|\cdot|_{\mathbb{C}}$ denotes norm in $\mathbb{C}$, $\chi\in\mathbb{F}$, and $\upsilon,\omega\in\mathcal{X}$.
\\
\linebreak
 The input product function $\xi_\emptyset^\otimes$ may be given in view of (\ref{fun}) as
 \[
 \xi_\emptyset^\otimes\equiv{\xi^-_\emptyset}^\otimes\otimes{\xi^+_\emptyset}^\otimes=\delta_\emptyset\otimes 1^\otimes,
 \]
 where $\delta_\emptyset=0^\otimes$ is vacuum vector $\delta_\emptyset(\upsilon)=0$ if $\upsilon\neq\emptyset$ in $\mathcal{X}$ and $\delta_\emptyset(\emptyset)=1$, and $1^\otimes$ is the identity in $L^\infty(\mathcal{X})$. It is therefore  referred to as the \emph{pseudo-vacuum} vector, and corresponds to an empty past. The pseudo-vacuum vector may also be considered as an operator $\boldsymbol{F}^\ddag_1:\mathfrak{h}\rightarrow\mathbb{H}$ embedding the object into the compound space $\mathbb{H}:=\mathfrak{h}\otimes\mathbb{F}$, such that
 \begin{equation}
 \boldsymbol{F}_1^\ddag\eta=\eta\otimes\xi_\emptyset^\otimes\equiv \delta_\emptyset\otimes\eta\otimes 1^\otimes.
 \end{equation}
The operator $\boldsymbol{F}^\ddag_1$ is called the pseudo-vacuum embedding \cite{Be92b}, and  it is the canonical isometric embedding  of the object into this quantum time-field (the second quantization of the clock). That is an environment of time particles, the quanta of the clock, that are all prepared in the future pure-state by the canonical embedding.

The second quantized  It\^o-Schr\"odinger equation may be given on variable chains $\tau\in\mathcal{X}$ as
\begin{equation}
\mathrm{d}\psi_{t}(\tau)=\boldsymbol{\gamma}(t)\psi_{t}(\tau)\mathrm{d}n_{t}(\tau), \quad \psi_0=\eta\otimes\xi_\emptyset^\otimes\equiv\boldsymbol{F}_1^\ddag\eta\label{bbt}
\end{equation}
where $n_t(\tau)=|\tau^t|$ is the number of interactions up to time $t$ such that $\tau^t=\tau\cap\mathbb{R}_+^t$ where $\mathbb{R}^t_+:=\{x\in\mathbb{R}_+:x<t\}\equiv[0,t)$, and $\boldsymbol{\gamma}(x)=\boldsymbol{\sigma}(x)-\mathbf{I}$ $\big($see (\ref{gamma})$\big)$ is the differential operator acting trivially as identity in $\underset{{z\in\tau\setminus x}}{\otimes}\Bbbk_z$.  The solution of (\ref{bbt}) has the form of sequential interactions parameterized by $t$, such that
\begin{equation}
\psi_{t}(\tau)=\overleftarrow{\boldsymbol{\sigma}}^\odot_t(\tau)\psi_0(\tau),\quad \overleftarrow{\boldsymbol{\sigma}}^\odot_t(\tau) :=\underset{x\in\tau^t}{\overset{\leftarrow}{\odot}}\boldsymbol{\sigma}(x),\label{sol}
\end{equation}
where $\odot$ is a chronological \emph{semi}-tensor product \cite{Be92,Be92b} well defined for $x_i<x_{i+1}$  as
\[
\boldsymbol{\sigma}(x_{i+1})\odot\boldsymbol{\sigma}(x_i)= \big(\boldsymbol{\sigma}(x_{i+1})\otimes\mathbf{I}_i\big)\;\big(\mathbf{I}_{i+1} \otimes\boldsymbol{\sigma}(x_i)\big)
\]
such that the composition $\boldsymbol{\sigma}(x_{i+1})\odot\boldsymbol{\sigma}(x_i)$ behaves as tensor product with respect to the clock and non-commutative chronological product in $\mathfrak{h}$. Notice that the action of $\overleftarrow{\boldsymbol{\sigma}}^\odot_t$ on tensor functions $\psi\in\mathbb{H}$ is diagonal so that $\big[\overleftarrow{\boldsymbol{\sigma}}^\odot_t\psi\big](\tau)=\overleftarrow{\boldsymbol{\sigma}}^\odot_t(\tau) \psi(\tau)$. Such diagonal operators are called \emph{graded}, and since $\overleftarrow{\boldsymbol{\sigma}}^\odot_t(\tau)$ is also decomposable on $\tau$ (\ref{sol}) it is called a graded product operator. When $\mathfrak{h}=\mathbb{C}$ the semi-tensor product $\odot$ simply coincides with the tensor product $\otimes$, then $\overleftarrow{\boldsymbol{\sigma}}^\odot_t(\tau)={\boldsymbol{\sigma}}^\otimes_t(\tau)= {\otimes}_{x\in\tau^t}\boldsymbol{\sigma}(x)$. Such graded tensor product functions may also be written using the second quantization functor as ${\boldsymbol{\sigma}}^\otimes=\Gamma({\boldsymbol{\sigma}})$.

Notice that the isometricity of $\boldsymbol{F}_1^\ddag$ is equivalent to the coherence of the pseudo-vacuum $\|\xi^\otimes_\emptyset\|=1$, and this isometrcity allows us to defines a conditional expectation $\mathbb{E}_\emptyset:\mathcal{B}(\mathfrak{h})\bar{\otimes}\mathbb{M}\rightarrow\mathcal{B}(\mathfrak{h})$  given on operators $\mathbf{X}\in\mathcal{B}(\mathfrak{h})\bar{\otimes}\mathbb{M}$ as
\begin{equation}
\mathbb{E}_\emptyset[\mathbf{X}]=\boldsymbol{F}_1\mathbf{X}\boldsymbol{F}_1^\ddag\label{ce}
\end{equation}
where  $\mathbb{M}=\Gamma(\mathfrak{m})$ is the graded Guichardet-Fock $\ddag$-group over $\mathfrak{m}$. That is the affine subspace $\oplus_{n=0}^\infty\mathfrak{m}^{\bar{\otimes} n}|\mathcal{X}$ of the block-diagonal operators in $\mathbb{F}$, where `block-diagonal' refers to operators that preserve the number of quanta of the clock.
\begin{theorem}
Let $\mathfrak{h}$ be a Hilbert space of square-integrable functions $\eta$ defining the probability amplitudes of the internal degrees of freedom of a quantum object,  and suppose that $\eta$ evolves in a unitary manner as a closed system  $\eta\mapsto\eta(t)$, $t>0$, satisfying the differential Schr\"odinger equation
\begin{equation}
\mathrm{d}\eta(t)=-\mathrm{i}H(t)\eta(t)\mathrm{d}t,\quad \eta(0)=\eta,\label{schr}
\end{equation}
for a time-dependent Hamiltonian $H(t)$, having the explicit solution
\begin{equation}
\eta(t)=V^t_0\eta,\quad V^t_s=\overleftarrow{\exp}\Big\{-\mathrm{i}\int^t_sH(x)\mathrm{d}x\Big\}\label{uprop}
\end{equation}
forming a two-parameter family $\{V^t_s\}$ satisfying the isometric {hemigroup property}
\[
V^t_sV^s_r=V^t_r,\;\;V^r_r=I,\quad t>s> r\in\mathbb{R}_+.
 \]
Then this two-parameter family in $\mathfrak{h}$ may be obtained as the conditional expectation (\ref{ce}) of the stochastic propagator  $\overleftarrow{\boldsymbol{\sigma}}^\odot_t\equiv\overleftarrow{\boldsymbol{\sigma}}^\odot_{[0,t)}$ describing the discrete interaction dynamics (\ref{bbt}), such that
\begin{equation} \mathbb{E}_\emptyset\big[\overleftarrow{\boldsymbol{\sigma}}^\odot_t\big]=V^t_0,
\end{equation}
given with respect to the pseudo-unitary interaction operator
\begin{equation}
\boldsymbol{\sigma}(x)=\left[
                         \begin{array}{cc}
                           I & -\mathrm{i}H(x) \\
                           0 & I \\
                         \end{array}
                       \right],\quad \boldsymbol{\sigma}(x)^\ddag=\boldsymbol{\sigma}(x)^{-1}\label{scat}
\end{equation}
entangling the object with the clock in the compound Minkowski-Hilbert space $\mathfrak{h}\otimes\Bbbk$.
\end{theorem}
\begin{proof}
Indeed one should first note that
\[
\xi_\emptyset^\ddag\boldsymbol{\sigma}_t\xi_\emptyset=-\mathrm{i}\int_0^tH(x)\mathrm{d}x,
\]
then it is left only to understand that
$
\mathbb{E}_\emptyset\big[\overleftarrow{\boldsymbol{\sigma}}^\odot_t\big] =\overleftarrow{\exp}\big\{\xi_\emptyset^\ddag\boldsymbol{\sigma}_t\xi_\emptyset\big\}
$
which follows from the fact that
\[
{\xi_\emptyset^\otimes}^\ddag \overleftarrow{\boldsymbol{\sigma}}^\odot_t \xi_\emptyset^\otimes= \int_{\mathcal{X}}{\xi_\emptyset^\otimes}^\ddag (\tau)\overleftarrow{\boldsymbol{\sigma}}^\odot_t(\tau) \xi_\emptyset^\otimes(\tau)\mathrm{d}\tau
\]
where
$
{\xi_\emptyset^\otimes}^\ddag (\tau)\overleftarrow{\boldsymbol{\sigma}}^\odot_t(\tau) \xi_\emptyset^\otimes(\tau)= \underset{{x\in\tau^t}}{\overleftarrow{\prod}}{\xi_\emptyset}^\ddag (x){\boldsymbol{\sigma}}(x) \xi_\emptyset(x)
$. Indeed the conditional expectation $\mathbb{E}_\emptyset[\cdot]$ is an averaging procedure over all possible numbers of interactions and all possible configurations for each fixed number of interactions, and it is in this way that it is in fact a Feynamn path integral.
\end{proof}
Note that the stochastic representation (\ref{bbt}) of the Schr\"odinger equation (\ref{schr}) may also be realized as a microscopic dilation of the unitary dynamics in $\mathfrak{h}$,
obtaining the unitary dynamics by averaging over the microscopic ensemble. In this case the microscopic ensemble is an opening up of time into its component particles; that is an ensemble of particles of temporal-spin. %In the case of time-space $\mathbb{R}_+$ microscopic refers purely to time, although a more general treatment was given in \cite{Be92b,thesis} where the time-space was an essentially ordered space $\mathbb{X}$.
\\
\linebreak
Historically, Belavkin had informally referred to the pseudo-Hilbert space of the general quantum stochastic dynamics  as the `heaven' space, obtained by the isometric embedding of the `earth' space of a general quantum stochastic dynamics. In this manner the \emph{earth projector} $\mathbf{E}=\boldsymbol{F}^\ddag_1\boldsymbol{F}_1$ is introduced which is a $\ddag$-orthoprojector $\mathbf{E}=\mathbf{E}^2=\mathbf{E}^\ddag$. In the case of Newton-Leibniz flow considered in this paper $\mathbf{E}$ projects the stochastic interaction dynamics  into the object space $\mathfrak{h}$, and then re-embeds the object into the object-clock space $\mathbb{H}$ as
\begin{equation}
\mathbf{E}:{\overleftarrow{\boldsymbol{\sigma}}}^\odot_t\boldsymbol{F}^\ddag_1\mapsto\boldsymbol{F}^\ddag_1 V^t_0.
\end{equation}

\section{Hamiltonian Dynamics as The Expectation of A Pseudo-Poisson Process}
So it has been shown above that the Feynman path integral of the open dynamics of a quantum system interacting with the quantum field of the clock gives rise to the evolution of a closed quantum system. Now we shall consider the Feynman path integral with respect to a Poisson measure as oppose to the standard Lebesgue measure. This means that we now consider the evolution of a closed quantum system to be obtained from a counting process of sequential object-clock interactions as above, but now it is supposed that the random interaction times $x\in\tau$ are distributed according to a Poisson probability law. Since the counting process is in a Minkowski-Hilbert space we refer to such a Poisson process as a pseudo-Poisson process.
\\
\linebreak
The solution (\ref{sol}) of the sequential interaction dynamics of the object with  the clock may be given in the single integral form as
\[
\psi_t(\tau)=\psi_0(\tau)+ \big[\mathbf{N}^t(\boldsymbol{\gamma})\overleftarrow{\boldsymbol{\sigma}}^\odot_t\psi_0\big](\tau)
\]
with respect to the counting process
\begin{equation}
\big[\mathbf{N}^t(\boldsymbol{\gamma})\overleftarrow{\boldsymbol{\sigma}}^\odot_t\psi_0\big](\tau)= \sum_{x\in\tau\cap[0,t)}\big[\boldsymbol{\gamma}(x)\odot \overleftarrow{\boldsymbol{\sigma}}^\odot_x(\tau\setminus x)\big]\psi_0(\tau),\label{ppp}
\end{equation}
where $\tau\setminus x=\{z\in\tau:z\neq x\}$.
One may re-consider this counting process as a pseudo-Poisson process distributed by the Poisson law
 \begin{equation}
\verb"P"^{}_\nu(\mathrm{d}\tau)=(2\nu)^{\otimes}(\tau)\mathrm{d} \tau e^{-\int_0^\infty2\nu(x)\mathrm{d}x},\label{plaw}
\end{equation}
such that the interaction times $\tau$ become Poisson events. The Poisson intensity $\nu$ describes the frequency of object-clock interaction, and the factor 2 arises from the dimensionality of $\Bbbk_x$.

The Guichardet-Fock representation  of  Minkowski-Poisson space is denoted by $\mathbb{H}_\nu$, and it is the space of $\verb"P"_\nu$ square-integrable functions $\varphi$, such that
\[
\|\varphi\|^2_\nu=\int_{\mathcal{X}}\|\varphi(\tau)\|^2\verb"P"_\nu(\mathrm{d}\tau)\in\mathbb{R},
\]
 and it is unitarily equivalent to the Minkowski-Hilbert space $\mathbb{H}$. In order to define functions $\varphi$ in $\mathbb{H}_\nu$ we must first
 note that the second quantization functor $\Gamma$ also admits the decomposition
 \[
 \Gamma\Big(L^1(\mathbb{R}_+)\oplus L^\infty(\mathbb{R}_+)\Big) =\Gamma\Big(L^1(\mathbb{R}^t_+)\oplus L^\infty(\mathbb{R}^t_+)\Big)\otimes\Gamma\Big(L^1(\mathbb{R}_t)\oplus L^\infty(\mathbb{R}_t)\Big)
 \]
 where $\mathbb{R}_+^t=[0,t)$ and $\mathbb{R}_t=[t,\infty)$, allowing us to decompose the Minkowski-Hilbert space at any time $t>0$ as
 \begin{equation}
 \mathbb{H}=\mathbb{H}^t\otimes\mathbb{F}_t,
 \end{equation}
 where $\mathbb{H}^0=\mathfrak{h}$ and $\mathbb{F}_0=\mathbb{F}$.
  Then we define the Minkowski-Poisson space  $\mathbb{H}_\nu^t$ as an isometry of the Minkowski-Hilbert space $\mathbb{H}^t$ for  a strictly positive intensity $\nu$ in $L^\infty(\mathbb{R}_+^t)$ such that both $\nu$ and $\nu^{-1}$ are integrable over $\mathbb{R}^t_+\subset\mathbb{R}_+$ for all $t\in\mathbb{R}_+\equiv\mathbb{R}_+^\infty$. This isometry is given by a unitary transformation of $\psi\in\mathbb{H}^t$ that is
\[
\varphi=\tfrac{1}{\sqrt{2\nu}}^\otimes\psi\; e^{\int_0^t\nu(x)\mathrm{d}x},
\]
 and the  Minkowski-Poisson space $\mathbb{H}_\nu$ may be  given  as the inductive limit $\mathbb{H}_\nu=\underset{t>0}{\cup}\mathbb{H}_\nu^t$, \cite{Be00b}.

Next consider an automorphism of the $\ddag$-algebra $\mathfrak{d}$ and $\ddag$-group $\mathfrak{m}=\mathbf{I}+\mathfrak{d}$ given by the $\ddag$-unitary real Lorentz transformation
 \begin{equation}
 \boldsymbol{\upsilon}_\lambda^\ddag=\left[
                     \begin{array}{cc}
                       \sqrt{\lambda} & 0 \\
                       0 & \sqrt{\tfrac{1}{\lambda}} \\
                     \end{array}
                   \right]
=\boldsymbol{\upsilon}_\lambda^{-1}\label{rl}
\end{equation}
that is well defined in $\Bbbk$ if $\lambda$ and its inverse are essentially bounded. In particular, we see that the vector $\xi_\lambda$ (\ref{basic}) in the restricted space $\Bbbk^t=\Bbbk|\mathbb{R}_+^t$ (of functions $\xi$ with support in $\mathbb{R}_+^t$) defines a vector $\frac{1}{\sqrt{2\lambda}}\xi_\lambda$ in the space $\Bbbk^t(\lambda)$ of $\lambda$-square integrable functions
\[
\|\xi\|^2_t(\lambda)=2\int^t_0\|\xi(x)\|^2\lambda(x)\mathrm{d}x,
\]
with, in particular,
\[
\tfrac{1}{\sqrt{2\lambda}}\xi_\lambda=\boldsymbol{\upsilon}^\ddag_\lambda\mathbf{p},\quad \mathbf{p}= \tfrac{1}{\sqrt{2}}\left[
             \begin{array}{c}
               1 \\
               1 \\
             \end{array}
           \right]
\]
 where $\mathbf{p}$ generates the canonical pseudo-Poisson state $\mathbf{p}^\otimes$ in $\mathbb{H}^t_\nu$. This pseudo-Poisson state is represented in the Minkowski-Hilbert space by the pseudo-Poisson embedding
 \begin{equation}
 {\boldsymbol{\Phi}^t_\nu}^\ddag=\sqrt{2\nu}^\otimes\mathbf{p}^\otimes e^{-\int_0^t\nu(x)\mathrm{d}x}
 \end{equation}
 defining the pseudo-Poisson expectation on the Minkowski-Hilbert space $\mathbb{H}^t$  as
\begin{equation}
\mathbb{P}^t_\nu[\mathbf{X}]=\boldsymbol{\Phi}^t_\nu\mathbf{X}{\boldsymbol{\Phi}_\nu^t}^\ddag,\label{pce}
\end{equation}
 where $\mathbf{X}\in\mathcal{B}(\mathfrak{h})\bar{\otimes}\mathbb{M}^t$ with $\mathbb{M}^t=\mathbb{M}|\mathcal{X}^t$. We may also write
\[
\mathbb{P}^t_\nu[\mathbf{X}]=\mathbb{P}_\nu[\mathbf{X}_t]\equiv \boldsymbol{\Phi}_\nu\mathbf{X}\boldsymbol{\Phi}_\nu^\ddag
\]
where $\mathbf{X}_t\in\mathcal{B}(\mathfrak{h})\bar{\otimes}\mathbb{M}$ is adapted which means $\mathbf{X}_t(\tau)=\mathbf{X}(\tau^t)\otimes\mathbf{I}^\otimes(\tau_t)$, with $\tau^t=\tau\cap\mathbb{R}^t_+$ and $\tau_t=\tau\cap\mathbb{R}_t$. Since if $\mathbf{X}_t$ is adapted in $\mathcal{B}(\mathfrak{h})\bar{\otimes}\mathbb{M}^r$ then $\mathbb{P}^r_\nu[\mathbf{X}_t]=\mathbb{P}^t_\nu[\mathbf{X}_t]\equiv\mathbb{P}^t_\nu[\mathbf{X}]$ for all $r\geq t$.

\begin{proposition}
Let $\nu$ be a strictly positive, bounded function over $\mathbb{R}^t_+$, then the propagator $V^t_0$  of the unitary evolution in $\mathfrak{h}$ (\ref{uprop}) may be obtained as the pseudo-Poisson expectation (\ref{pce}) of the stochastic propagator $\mathbf{S}_t:=\overleftarrow{{\boldsymbol{s}}}_t^\odot$ of the counting process $\mathbf{N}^t({\boldsymbol{l}}) \mathbf{S}_t$, where ${\boldsymbol{l}}(x)=\boldsymbol{s}(x)-\mathbf{I}$, such that
\begin{equation}
\mathbb{P}^t_\nu[\mathbf{S}]=V^t_0,
\end{equation}
where the interaction operator $\boldsymbol{\sigma}$ of (\ref{sol}) is given as a real Lorentz transformation of the interaction operator $\boldsymbol{s}$ such that $\boldsymbol{\sigma}=
\boldsymbol{\upsilon}^\ddag_\nu{\boldsymbol{s}}\boldsymbol{\upsilon}_\nu\equiv\boldsymbol{s}_\nu$, and $\nu$ is the pseudo-Poisson intensity describing the variable frequency of the object-clock interactions.
 \end{proposition}
\begin{proof}
The coherent vector $\xi_\nu^\otimes\exp\big\{-\int_0^t\nu(x)\mathrm{d}x\big\}\in\mathbb{H}^t$ may be obtained as a pseudo-Weyl transformation $\mathbf{W}_\nu\xi_\emptyset^\otimes$ of the pseudo-vacuum vector. The  $\ddag$-unitary pseudo-Weyl operator has the form
\[
\mathbf{W}(\xi)=\exp\left\{\mathbf{A}^\ddag(\xi)-\mathbf{A}(\xi^\ddag)\right\}
\]
given by the pseudo- creation and annihilation operators $\mathbf{A}^\ddag$ and $\mathbf{A}$, whose   action is given on tensor functions  $\psi\in\mathbb{H}^t$ as
\[
\big[\mathbf{A}^\ddag(\xi)\psi\big](\tau)=\sum_{x\in\tau}\xi(x)\otimes\psi(\tau\setminus x),\quad
\big[\mathbf{A}(\xi^\ddag)\psi\big](\tau)=\int_{0}^t\xi^\ddag(x)\psi(x\sqcup \tau)\mathrm{d}x,
\]
 where $\xi\in\Bbbk^t$. Thus the Weyl operator $\mathbf{W}_\nu$ is given by such $\mathbf{W}(\xi)$ with $\xi=(0,\nu)^\ddag$.

 The  pseudo-Poisson embedding  $\mathfrak{h}\rightarrow\mathbb{H}^t_\nu$ of the object Hilbert space into the Minkowski-Poisson space up to time $t$ may now be given as a $\ddag$-unitary transformation of the pseudo-vacuum embedding as
\begin{equation}
{\boldsymbol{\Phi}^t_\nu}^\ddag=\Upsilon_\nu\mathbf{W}_\nu{\boldsymbol{F}_1^t}^\ddag,\quad \Upsilon_\nu=\boldsymbol{\upsilon}^\otimes_\nu.
\end{equation}
 Then since the commutator $[\mathbf{W}_\nu,\Upsilon^\ddag_\nu\mathbf{X}\Upsilon_\nu]=0$ for all $\mathbf{X}\in\mathcal{B}(\mathfrak{h})\bar{\otimes}\mathbb{M}^t$,  following from the realization that the action of $\mathbf{X}$ leaves the second quantization of $\xi=(0,\nu)^\ddag\in\Bbbk^t$ invariant for all $t\in\mathbb{R}_+$, it follows that
\begin{equation}
\mathbb{P}^t_\nu[\mathbf{X}]=
\mathbb{E}_\emptyset\big[\Upsilon^\ddag_\nu\mathbf{X}_t\Upsilon_\nu\big],
\end{equation}
where $\mathbf{X}_t=\mathbf{X}\otimes\mathbf{I}^\otimes_{t}$ is the embedding of $\mathbf{X}$ into $\mathcal{B}(\mathfrak{h})\bar{\otimes}\mathbb{M}$, with $\mathbb{M}=\mathbb{M}^t\bar{\otimes}\mathbb{M}_t$ and $\mathbf{I}^\otimes_t$ is the identity operator in $\mathbb{M}_t$.
Finally we set $\Upsilon^\ddag_\nu\mathbf{X}_t\Upsilon_\nu={\overleftarrow{{\boldsymbol{\sigma}}}}^\odot_t$,  which gives us the required form  $\mathbf{X}=\mathbf{S}\equiv{\overleftarrow{{\boldsymbol{s}}}}^\odot\in\mathcal{B}(\mathfrak{h}) \bar{\otimes}\mathbb{M}^t$ for all $t>0$, where $\boldsymbol{s}=\boldsymbol{\upsilon}_\nu\boldsymbol{\sigma}\boldsymbol{\upsilon}_\nu^\ddag$.
\end{proof}
\begin{remark}
Notice that if the Hamiltonian  has separable time dependence $H(t)=H\nu(t)$ then the interaction operator $\boldsymbol{s}$, defined over $\mathbb{R}_+^t$ for all $t>0$, has the form
\[
\boldsymbol{s}=\left[
                 \begin{array}{cc}
                   I & -\mathrm{i}H\otimes1 \\
                   0 & I \\
                 \end{array}
               \right]
\]
noting that $1\in L^\infty(\mathbb{R}^t_+)\subset L^1(\mathbb{R}^t_+)$.

Also note that the isometric embedding ${\boldsymbol{\Phi}^t_\nu}^\ddag$  of $\mathfrak{h}$ into the Minkowski-Poisson space $\mathbb{H}^t_\nu$ defines the  $\ddag$-orthoprojector $
 {{\boldsymbol{\Phi}}^t_\nu}^\ddag\boldsymbol{\Phi}^t_\nu=\mathbf{P}^t_{\nu}$ where
 \[\mathbf{P}^t_\nu:\mathbf{S} {\boldsymbol{\Phi}^t_\nu}^\ddag\mapsto{\boldsymbol{\Phi}^t_\nu}^\ddag V^t_{0}.\]
 Alternatively we may consider the Lorentz-transformed earth projector $\mathbf{E}^t_\nu=\Upsilon_\nu\mathbf{E}^t\Upsilon^\ddag_\nu\equiv{\boldsymbol{F}^t_\nu}^\ddag\boldsymbol{F}^t_\nu$ given with respect to scaled pseudo-vacuum embedding ${\boldsymbol{F}^t_\nu}^\ddag=\Upsilon_\nu{\boldsymbol{F}_1^t}^\ddag$, since
 \begin{equation}\boldsymbol{\Phi}^t_\nu\mathbf{S}{\boldsymbol{\Phi}^t_\nu }^\ddag =\boldsymbol{F}^t_{\nu}\mathbf{S}{\boldsymbol{F}^t_{\nu}}^\ddag,\end{equation}
 such that the pseudo-Poisson expectation may be given as the scaled pseudo-vacuum expectation $\mathbb{E}_\nu^t[\mathbf{S}]=\mathbf{E}^t_\emptyset\big[\Upsilon^\ddag_\nu\mathbf{S}\Upsilon_\nu\big]$.
\end{remark}

\section{The Schr\"odinger Picture and The Dirac Equation}
In this chapter we shall extend the time domain $\mathbb{R}_+$ into the negative domain as a mathematical convenience. For this extension allows us to fix the quantum object at the origin of the time domain with the events now propagating towards the object from its future without any loss of the output information (the memory) that now flows into the negative time domain. This vision is called the Schr\"odinger picture and introduces a propagation of the quantum clock wave-functions, $\xi$, generated by  the \emph{momentum of the clock} \cite{Be00a,Be00b}. This leads a free dynamics of the clock in the extended Minkowski-Fock space $\widetilde{\mathbb{F}}\otimes\mathbb{F}$ in addition to the otherwise purely stochastic dynamics of the clock interacting with the quantum object that has been considered thus far. The Schr\"odinger picture gives a further insight into the microscopic dynamics of a quantum system. It provides a physical picture in which objects remain present whilst time flows through them, and this flow of time has momentum. Meanwhile, this present object interacts with this flow of time at random times. These random object-clock interactions are reckoned to occur with a frequency called Poisson intensity (introduced in the previous chapter), and this frequency can become lesser or greater as the time continues to flow through the object. The interactions of the present object with the flow of time transforms the state of time into past, generating a history of memory that is entangled with the present object.
\\
\linebreak
The Schr\"odinger picture gives rise to the second quantized Dirac equation as we shall see below. The construction of the Schr\"odinger picture shall begin with the definition of the shift operator $\boldsymbol{u}^t=\exp\{t\partial_x\}$ in the quantum clock  space $\Bbbk\equiv L^1(\mathbb{R}_+)\oplus L^\infty(\mathbb{R}_+)$,  with action  defined on the functions $\xi\in\Bbbk$ as
\begin{equation}
\big[\boldsymbol{u}^t\xi\big](x)\equiv\xi^t(x) =\xi(x+t),\quad t>0.\label{shift}
 \end{equation}
However, this shift is not $\ddag$-unitary but only $\ddag$-coisometry $\boldsymbol{u}^t{\boldsymbol{u}^t}^\ddag=\mathbf{I}$ inducing the projection ${\boldsymbol{u}^t}^\ddag \boldsymbol{u}^t$  onto  $\mathbb{R}_t=\mathbb{R}_+\setminus[0,t)$ such that $\big[{\boldsymbol{u}^t}^\ddag\boldsymbol{u}^t\xi\big](x)=0$ for all $x<t$. This may be resolved by extending the time-space into the negative domain $\mathbb{R}_+\mapsto\mathbb{R}_-\sqcup\mathbb{R}_+$ such that
the vector functions $\xi$ now live in the extended Minkowski-Hilbert space $\mathbb{K}=\widetilde{\Bbbk}\oplus\Bbbk$, where  $\widetilde{\Bbbk}=L^1(\mathbb{R}_-)\oplus L^\infty(\mathbb{R}_-)\cong\Bbbk$ and $\mathbb{R}_-=\widetilde{\mathbb{R}}_+\equiv-\mathbb{R}_+$, such that $\widetilde{x}=-x$, then $\boldsymbol{u}^t$ is $\ddag$-unitary in $\mathbb{K}$ for all $t>0$. $\mathbb{R}_-\sqcup\mathbb{R}_+$ is called the disjoint union of $\mathbb{R}_+$ and $\mathbb{R}_-$, and may be understood as the real line $\mathbb{R}$ with a degenerate origin $(\widetilde{0},0)$.

In the \emph{interaction picture} (\ref{bbt}) of the previous chapters  one may conceive that the object  is propagating into its future along the time-space $\mathbb{R}_+$, and interacting with the clock at arbitrary times $x\in\tau$ where future is transformed into past. Extending the time domain into the past does not alter the mechanics provided that $\boldsymbol{\sigma}(x):=\mathbf{I}$ for all $x\in\mathbb{R}_-$, given that $\xi_\emptyset(\widetilde{x}):=\xi_\emptyset({x})$, but this initial configuration is not a unique choice. However,
the extension allows us to transform to the \emph{Schr\"odinger picture} where the wave-function of the clock now propagates backwards upon $\mathbb{R}$ under the action of the shift $\boldsymbol{u}^t$ (\ref{shift}). Consequently, an object-clock interaction does not occur at some $x>0$ as in the interaction picture, but always at the degenerate origin $(\widetilde{0},0)$ of the time domain where the clock wave-function becomes discontinuous as a result of interaction with the object.  In the Schr\"odinger picture the object no longer moves into its future. Instead, the future comes to the object, and with this flow of time comes the momentum of the clock \cite{Be00a,Be00b}.
\begin{proposition}
 Let  $\xi^t_\emptyset\in\mathbb{K}$  denote  the extended propagating canonical input wave $\boldsymbol{u}^t\xi_\emptyset$,
and let ${\psi}^t=\mathbf{U}^t\psi_t$, where $\psi_t$ is the solution of {(\ref{bbt})} and $\mathbf{U}^t=\mathbf{I}^\otimes\exp\{t\partial_{\varkappa}\}$ is the second quantized unitary shift operator and $\partial_\varkappa=\sum_{x\in\varkappa}\partial_{x}$ is the generator of this free-shift dynamics of the clock. Then the stochastic It\^o-Schr\"odinger equation {(\ref{bbt})} is $\ddag$-unitarily equivalent to the pseudo-selfadjoint Schr\"odinger boundary value problem in $\widetilde{\mathbb{H}}\otimes\mathbb{F}$, where $\widetilde{\mathbb{H}}=\Gamma(\widetilde{\Bbbk})\otimes\mathfrak{h}$, given as
\begin{equation}
\mathrm{i}\partial_{t}{\psi}^{t}(\varkappa)=\boldsymbol{\mu}
{\psi}^{t}(\varkappa),\quad {\psi}^{t}(\varkappa\sqcup\widetilde{0}) =\boldsymbol{\sigma}(t){\psi}^{t}(\varkappa\sqcup0),\label{sp}
\end{equation}
where  $\boldsymbol{\mu}=\mathrm{i}\partial_{\varkappa}= \boldsymbol{\mu}^\ddag$ is the sum of  momenta of the clock, $\boldsymbol{\sigma}(t)$ is the canonical pseudo-unitary object-clock interaction operator (\ref{scat}), and $\varkappa\subset\mathbb{R}$.

Moreover, this Schr\"odinger boundary value problem may be written as the second quantized `mometumless' Dirac equation \cite{Me62} in $\mathbb{C}^2\otimes\mathbb{H}$, with boundary condition in $\mathbb{H}$, that is
\begin{equation}
\left[
  \begin{array}{cc}
    \partial_t+\partial_\tau & 0 \\
    0 & \partial_t-\partial_\tau \\
  \end{array}
\right]\left[
         \begin{array}{c}
           {{\psi}}_-^t \\
           {\psi}_+^t \\
         \end{array}
       \right]=0,\quad {{\psi}}_-^{t}({0}) =\boldsymbol{\sigma}(t){\psi}_+^{t}(0),\label{de}
\end{equation}
where  ${{\psi}}_-^t(\tau)={\psi}^t(\widetilde{\tau})\equiv[\mathbf{R}\psi^t](\tau)$  is the reflection of the part of $\psi^t$ that is restricted to the negative time domain, and $\psi^t_+(\tau)=\psi^t(\tau)$ is just the restriction of $\psi^t$ to the positive time domain. $\tau\subset\mathbb{R}_+$.
\end{proposition}
\begin{proof}
  Following closely to Belavkin \cite{Be00a} we note that the boundary value problem (\ref{sp}) is pseudo-selfadjoint on the extended domain of $\widetilde{\mathbb{H}}\otimes\mathbb{F}$-valued functions, absolutely continuous at non-zero $x\in\varkappa$, and right-continuous at any $0\in\varkappa$ satisfying also the pseudo-unitary boundary condition. Indeed the pseudo-unitarity of both $\mathbf{U}^t$ and $\overleftarrow{\boldsymbol{\sigma}}^\odot_t$ implies the pseudo-unitarity of the Schr\"odinger propagator $\mathbf{V}^t:=\mathbf{U}^t\overleftarrow{\boldsymbol{\sigma}}^\odot_t$. Moreover, the map $t\mapsto\mathbf{V}^t$ has the multiplicative representation property $\mathbf{V}^r\mathbf{V}^t=\mathbf{V}^{r+t}$ of the semi-group $\mathbb{R}_+\ni r,t$ because the map $t\mapsto\boldsymbol{\sigma}_t(x)$ is a multiplicative $\boldsymbol{u}^t$-cocycle such that
 \[
 {\boldsymbol{u}^t}^\ddag\boldsymbol{\sigma}_r \boldsymbol{u}^t\boldsymbol{\sigma}_t= \boldsymbol{\sigma}_{r+t}.
 \]

 The Dirac equation follows naturally from a reflection of the negative time domain so that we have two positive time domains, one for input and the other for output. Notice that we may define $\psi^t_{\pm}=\Psi^t_{\pm}\eta$, where $\Psi^t_\pm\in\mathcal{L}(\mathfrak{h},\mathbb{H})$, that is the space of adjointable maps from $\mathfrak{h}$ into $\mathbb{H}$, such that for any $\varkappa=\widetilde{\tau}\sqcup\tau$, $\tau\subset\mathbb{R}_+$, we have $\psi^t(\varkappa)= \Big[ [\mathbf{R}\Psi_-^t](\widetilde{\tau})\odot\Psi_+^t(\tau)\Big]\eta$.
\end{proof}
We may also present this Schr\"odinger boundary value problem in the extended Minkowski-Poisson space $\widetilde{\mathbb{H}}^r_\nu\otimes\mathbb{F}^r_\nu$ if the interaction intensity $\nu$ is strictly positive, bounded, \emph{and smooth} function over $\mathbb{R}_-^r\sqcup\mathbb{R}_+^r$, where $\mathbb{R}^r_{-}=(-r,0]\equiv -\mathbb{R}^r_+$. It has a similar form to (\ref{sp}) but with an additional potential $\phi(\varkappa)=\sum_{x\in\varkappa}\phi(x)$, and may be written as
\begin{equation}
\mathrm{i}\partial_{t}{\varphi}^{t}(\varkappa)=\boldsymbol{\mu}_\phi(\varkappa,\mathrm{i} \partial_\varkappa)
{\varphi}^{t}(\varkappa),\quad {\varphi}^{t}(\varkappa\sqcup\widetilde{0}) ={\boldsymbol{s}}(t){\varphi}^{t}(\varkappa\sqcup0),\label{pbvp}
\end{equation}
as it was done similarly in \cite{Be00a} (although that was neither second quantized nor in Minkowski space, but the underlying principles are the same). The operator
\begin{equation}\boldsymbol{\mu}_\phi(\varkappa,\mathrm{i} \partial_\varkappa)=\mathrm{i}\big(\phi(\varkappa)+\partial_\varkappa\big)\equiv \tfrac{1}{\sqrt{\nu}}^\otimes(\varkappa)\mathrm{i}\partial_\varkappa\sqrt{\nu}^\otimes(\varkappa)\end{equation}
 is the sum of the clock momenta in the potential, and $\varphi^t={\tfrac{1}{\sqrt{2\nu}}}^\otimes\psi^t\;e^{\int^r_{-r}\nu(x)\mathrm{d}x}$, $r\geq t$, is the unitary transformation of $\widetilde{\mathbb{H}}^r\otimes\mathbb{F}^r$ into $\widetilde{\mathbb{H}}^r_\nu\otimes\mathbb{F}^r_\nu$. The potentials $\phi(x)$ generate the clock-object interaction amplitudes $\sqrt{\nu}$ so that they are given as the derivatives $\phi(x)=\partial_x\vartheta(x)$ of the hyperbolic angles $\vartheta(x)=\tfrac{1}{2}\ln\nu(x)$ associated with the real Lorentz transforms $\boldsymbol{\upsilon}_\nu(x)$ at each $x$ (\ref{rl}).
Indeed the boundary value problem (\ref{pbvp}) is unitarily equivalent to (\ref{sp}), and therefore also to (\ref{bbt}), and so the Feynman path integral of this Schr\"odinger-Poisson boundary value problem of an object interaction with a flow of time in a potential is indeed Hamiltonian dynamics of the quantum object. In this case we the expectation should be the Feynamn-Poisson path integral as it is given with respect to a Poisson measure with intensity $\nu$.

In the Schr\"odinger picture the time-waves $\xi^t$ continuously flow through the object  with   momenta given by the operation of $\boldsymbol{\mu}$. When the object is fixed at the origin of the time domain its evolution is generated by  its entanglement with the otherwise freely propagating time-waves. The entanglements occur at random times $t\in\tau$ corresponding to the zeros in $\tau-t$, where $\tau$ is an initial configuration of potential interactions.  If the time-waves interact with the object they  continue to propagate into the past but the object's state $\eta$ is transformed by the action of the generator $G(x)=-\mathrm{i}H(x)$ $\big($see (\ref{pui},\ref{act})$\big)$, and this `pseudo-reduction' of the object state $\eta$  may be understood as an indirect measurement of the passing of time; that is the transition of future into past. As the interactions continue the object comes to be evolved under the action of the chronological exponent $G(\varkappa+t)=\underset{{x\in\varkappa}}{\overleftarrow{\prod}}G(x+t)$, where $\varkappa\subset\mathbb{R}_-^t$.

The Dirac equation is similar, but considers the object at the origin of a \emph{radial} time-space with input-output degeneracy corresponding to the two degrees of freedom of the clock. Thus we have a space of  incoming flux of time particles and a space of outgoing flux of time particles. The input is the future and the output is the past. The object connects these two independent channels and is entangled with the latter.

When studying the dynamics in the extended Minkowski-Poisson space we were able to obtain the Hamiltonian dynamics in $\mathfrak{h}$ from a Poisson expectation of the Schr\"odinger-Poisson dynamics, which is given with respect to the Poisson law (\ref{plaw}). Note that when speaking of a Poisson intensity $\nu$ we have a Poisson measure with a factor $2\nu$. This may in fact be understood as the sum $\nu_-+\nu_+$ of independent intensities of the two channels; this means that $\nu_+$ is the intensity of the input channel and $\nu_-$ is the intensity of the output channel. Generally these intensities need not be the same. The intensity of the flow of time into an object may be greater or lesser than intensity of the flow of time out of an object.
\\
\linebreak
The remainder of this chapter is optional and it considers an alternative momentum of the clock. In the above construction the clock momentum generates a flow of both the past and future degrees of freedom of the clock towards the present quantum object from the future. In the construction below the considered clock momentum generates a flow of the future degree of freedom of the clock in one direction and of the past degree of freedom of the clock in the opposite direction.

This begins with a consideration of the Schr\"odinger boundary value problem in $\widetilde{\mathbb{H}}\otimes\mathbb{F}$ with alternative time-wave momentum   \begin{equation}\boldsymbol{\zeta}=\sum_{x\in\varkappa}
  \left[
                         \begin{array}{cc}
                           -\mathrm{i}\partial_x & 0 \\
                           0 & \mathrm{i}\partial_x \\
                         \end{array}
                       \right]\otimes \mathbf{I}^\otimes_{\varkappa\setminus x},\end{equation}
  $\varkappa\subset\mathbb{R}$,
  such that the two degrees of freedom of the clock  are considered to propagate in opposite directions over $\mathbb{R}$. Thus we shall consider a {partial} reflection of the boundary value problem (\ref{sp}) with respect to the graded product operator
\[
\mathbf{J}=\left[
             \begin{array}{cc}
               R & 0 \\
               0 & I \\
             \end{array}
           \right]^\otimes,\quad \mathbf{J}^\ddag\mathbf{J}=\mathbf{R}.
\]
Here $R$ is the unitary clock-coordinate reflection operator $\big[RG^t\big](x)=G^t(\widetilde{x})$,
 where $G^t(x)=G(x+t)$, and $\mathbf{J}=\mathbf{J}^\ast=\mathbf{J}^{-1}$.  In this picture the Schr\"odinger boundary value problem becomes
\begin{equation}
\mathrm{i}\partial_t{\psi}_{\mathbf{J}}^t(\varkappa)= \boldsymbol{\zeta}{\psi}_{\mathbf{J}}^t(\varkappa),\quad {\psi}_{\mathbf{J}}^t(\varkappa\sqcup\widetilde{0}) ={\boldsymbol{\varsigma}}(t){\psi}_{\mathbf{J}}^t(\varkappa\sqcup0),\label{sp2}
\end{equation}
where ${\psi}^t_{\mathbf{J}}=\mathbf{Z}^t\mathbf{J}\overleftarrow{\boldsymbol{\sigma}}^\odot_t\psi_0 \equiv\mathbf{J}\mathbf{U}^t\overleftarrow{\boldsymbol{\sigma}}^\odot_t \psi_0$ where  $\mathbf{Z}^t$ is the second quantized shift generated by $\boldsymbol{\zeta}$, the sum of clock momenta, and $\boldsymbol{\varsigma}=\mathbf{J}\boldsymbol{\sigma}\mathbf{J} \equiv\boldsymbol{\sigma}_{\mathbf{J}}$. This shift is not $\ddag$-unitary but we naturally introduce the pseudo-involution $\mathbf{Z}^\dag:=\mathbf{J}{\big(\mathbf{JZJ}\big)}^\ddag\mathbf{J}\equiv\boldsymbol{R}\mathbf{Z}^\ast\boldsymbol{R}$, such that the $\dag$-involution is induced by the pseudo-metric $\boldsymbol{R}=\mathbf{J}\sigma^\otimes_1\mathbf{J}\equiv\mathbf{R}\sigma_1^\otimes$. Indeed we see that $\boldsymbol{\zeta}^\dag=\boldsymbol{\zeta}$ and $\mathbf{Z}$ is $\dag$-unitary.

The restriction ${\psi}^t_{\mathscr{R}}=\mathbf{J}{\psi}^t|\mathcal{X}$ of ${\psi}_\mathbf{J}^t$ to the positive simplex $\mathcal{X}$, is an element of $\mathbb{H}\equiv L^1(\mathcal{X})\otimes\mathfrak{h}\otimes L^\infty(\mathcal{X})$ having the component functions
\begin{equation}
{\psi}^t_{\mathscr{R}}(\upsilon_t,\omega_t)=\big[U^{-t}R^\otimes G^\odot\big](\upsilon_t)\eta\otimes \big[U^t1\big](\omega_t),\label{prop}
\end{equation}
where $\upsilon_t,\omega_t$ are respectively the restrictions of $\widetilde{\tau}+t,\tau-t$ to the positive domain $\mathcal{X}\ni\tau$, and $U^t=\exp\{t\partial_{\tau}\}$ satisfying $U^{-t}R=RU^t$.
\begin{remark}
 Note that the function (\ref{prop}) describes a chronological exponent of the generators $G$ composing the object's history, and propagates upon the outgoing  chain $\upsilon_t$. Similarly,  the future oriented temporal spins propagate towards the object from its future upon the incoming chain $\omega_t$. Also note that for a fixed $\tau\in\mathcal{X}$ the cardinality of $\upsilon_t\sqcup\omega_t\cong\tau$ is fixed, but as $t$ increases $|\upsilon_t|$ increases and $|\omega_t|$ decreases.
\end{remark}

Now we consider the total reflection $\mathbf{R}=\mathbf{J}^\ddag\mathbf{J}$ of the restriction ${\psi}_{\mathbf{J}}^t|\widetilde{\mathcal{X}}$, denoted by ${\psi}^t_{\mathscr{I}}=\mathbf{R}{\psi}_{\mathbf{J}}^t|{\mathcal{X}}$. The functions ${\psi}^t_{\mathscr{R}}$ and ${\psi}^t_{\mathscr{I}}$ are now both elements of $\mathbb{H}$, but we would like to consider the vector operators ${\Psi}^t\in\mathcal{L}(\mathfrak{h},\mathbb{H})$ with the property ${\psi}^t_{\mathscr{R}}={\Psi}^t_{\mathscr{R}}\eta$ and ${\psi}^t_{\mathscr{I}}= {\Psi}^t_{\mathscr{I}}\eta$. Then we find that in the co-space $\mathcal{L}(\mathbb{H},\mathfrak{h})$ we may consider the $\mathcal{B}(\mathfrak{h})$-valued linear functional ${\Psi}^{t\;\flat}_{\mathscr{I}}$ that maps the dynamical information ${\Psi}^{t}_{\mathscr{R}}$ into the deterministic unitary dynamics of the object in $\mathfrak{h}$ such that
\begin{equation}
{{\Psi}^{t\;\flat}_{\mathscr{I}}}{\Psi}^{t}_{\mathscr{R}}=V^t_0,
\end{equation}
  where the $\flat$-involution is induced by the pseudo-metric $\mathrm{i}\sigma_2\equiv\pi(\mathrm{d}t)-\pi(\mathrm{d}t)^\ast$. Indeed we may  construct a second quantized `massless' Dirac equation \cite{Me62} with boundary condition, corresponding to the Schr\"odinger boundary value problem (\ref{sp2}),
\[
\left[
  \begin{array}{cc}
    \mathrm{i}\partial_t-\boldsymbol{\zeta} & 0 \\
    0 & \mathrm{i}\partial_t+\boldsymbol{\zeta} \\
  \end{array}
\right]
\left[
    \begin{array}{c}
      {\psi}^{t}_{\mathscr{I}} \\
      {\psi}^{t}_{\mathscr{R}} \\
    \end{array}
  \right]=0,\quad{\psi}^{t}_{\mathscr{R}}(0)={\boldsymbol{\varsigma}}(t) {\psi}^{t}_{\mathscr{I}}(0).
  \]

\section{Regarding The General Theory}
The general theory of quantum stochastic calculus may be handled in a similar manner to this presentation of the deterministic dynamics of a closed system. The most basic example of quantum stochastic calculus works with operators in the Guichardet-Fock space $\mathcal{F}=\Gamma(\mathfrak{k})$ where $\mathfrak{k}=L^2(\mathbb{R}_+)$ is the space of square-integrable functions over $\mathbb{R}_+$ and $\Gamma$ is the second quantization functor. Using the Belavkin formalism we can replace the Hilbert space $\mathfrak{h}$ considered above with the Guichardet-Fock space $\mathcal{F}:=\Gamma\big(L^2(\mathbb{R}_+)\big)$, in which case we no longer have a Schr\"odinger equation but instead a quantum stochastic differential equation. However, one may proceed as it is done here in the deterministic case and compose the space $\mathcal{F}$ with the Minkowski-Fock space $\mathbb{F}=\Gamma\big(L^1(\mathbb{R}_+)\big)\otimes \Gamma\big(L^\infty(\mathbb{R}_+)\big)$ of the clock. Thus we obtain the dilated Minkowski-Fock space
\[
\mathcal{F}\otimes\mathbb{F}=\Gamma\big(L^1(\mathbb{R}_+)\oplus L^2(\mathbb{R}_+) \oplus L^\infty(\mathbb{R}_+)\big) :=\mathfrak{F}.
\]
In this way Belavkin was able to study quantum stochastic processes in terms of the fundamental pseudo-Poisson processes in the dilated Minkowski-Fock space $\mathfrak{F}=\Gamma(\mathfrak{K})$, where
\[
\mathfrak{K}=L^1(\mathbb{R}_+)\oplus L^2(\mathbb{R}_+) \oplus L^\infty(\mathbb{R}_+)\equiv\mathfrak{k}^-\oplus \mathfrak{k}^\circ \oplus\mathfrak{k}^+.
\]
He discovered that quantum stochastic processes may be obtained from fundamental pseudo-Poisson processes having precisely the graded product form considered in this paper, but with an additional noise degree of freedom of the clock corresponding to the additional space $\mathfrak{k}^\circ=L^2(\mathbb{R}_+)$. That is the space of a single particle of quantum noise, and may indeed be considered as a more general space than $L^2(\mathbb{R}_+)$.

This basically means that the operators in the Guichardet-Fock space $\mathcal{F}$, quantum stochastic processes, may be given as the conditional expectations $\mathbb{E}_\emptyset\big[\overleftarrow{\boldsymbol{\sigma}}^\times_t\big]$ (compressing the second quantized clock) of the graded product stochastic
propagators $\overleftarrow{\boldsymbol{\sigma}}^\times_t$ of counting processes, with the interaction operators $\boldsymbol{\sigma}$ now have the general form
\[
\boldsymbol{\sigma}=\left[
                      \begin{array}{ccc}
                        1 & \sigma^-_\circ & \sigma^-_+ \\
                        0 & \sigma^\circ_\circ & \sigma^\circ_+ \\
                        0 & 0 & 1 \\
                      \end{array}
                    \right],
\]
where $\sigma^-_+$ is the integrable deterministic generator, $\sigma^-_\circ$ is the annihilation functional on the noise space $\mathfrak{k}^\circ$, $\sigma^\circ_+$ is the creation vector function in $\mathfrak{k}^\circ$, and $\sigma^\circ_\circ$ is scattering operator that generates the counting dynamics; it operates in $\mathfrak{k}^\circ_t$ at each time $t$, and $\sigma^\circ_\circ\in L^\infty(\mathbb{R}_+)$ if $\mathfrak{k}^\circ= L^2(\mathbb{R}_+)$. In the absence of noise there are no $\circ$-terms, and we recover the structures considered above with $H=\mathrm{i}\sigma^-_+$.

%\begin{acknowledgements}
%If you'd like to thank anyone, place your comments here
%and remove the percent signs.
%\end{acknowledgements}

% BibTeX users please use one of
%\bibliographystyle{spbasic}      % basic style, author-year citations
%\bibliographystyle{spmpsci}      % mathematics and physical sciences
%\bibliographystyle{spphys}       % APS-like style for physics
%\bibliography{}   % name your BibTeX data base

% Non-BibTeX users please use

\end{document}